\documentclass{article}

\usepackage{amsthm}
\usepackage{amssymb}
\usepackage{amsfonts}
\usepackage{times}
\usepackage{url}
\usepackage{booktabs}
\usepackage{longtable}
\usepackage{epic,eepic}
\usepackage{amsmath}

\newtheorem{lemma}{Lemma}
\newtheorem{theorem}{Theorem}

\newtheorem{proposition}{Proposition}

\newcommand{\2}{\vspace{2mm}}
\newcommand{\be}{\begin{enumerate}}
\newcommand{\ee}{\end{enumerate}}
\newcommand{\dom}{\mbox{$\rightarrow$}}

\begin{document}

\title{Fixed-Parameter Algorithms in Analysis of Heuristics for Extracting Networks in Linear Programs\footnote{A preliminary version of this paper will appear in the Proceedings of the 4th International Workshop on Parameterized and Exact Computation (IWPEC'09).}}
\author{
Gregory Gutin\footnote{Department of Computer Science,
Royal Holloway, University of London, Egham, Surrey TW20 0EX,
England, UK, \texttt{gutin@cs.rhul.ac.uk}}\and Daniel Karapetyan\footnote{Department of Computer Science,
Royal Holloway, University of London, Egham, Surrey TW20 0EX,
England, UK, \texttt{daniel.karapetyan@gmail.com}} \and Igor Razgon\footnote{Department of Computer Science, University College Cork, Ireland, \texttt{i.razgon@cs.ucc.ie}}}

\maketitle
\begin{abstract}
We consider the problem of extracting a maximum-size reflected network in a linear program.  This problem has been studied before and a state-of-the-art SGA heuristic with two variations have been proposed.

In this paper we apply a new approach to evaluate the quality of SGA\@.  In particular, we solve majority of the instances in the testbed to optimality using a new fixed-parameter algorithm, i.e., an algorithm whose runtime is polynomial in the input size but exponential in terms of an additional parameter associated with the given problem.

This analysis allows us to conclude that the the existing SGA heuristic, in fact, produces solutions of a very high quality and often reaches the optimal objective values.  However, SGA contain two components which leave some space for improvement: building of a spanning tree and searching for an independent set in a graph.  In the hope of obtaining even better heuristic, we tried to replace both of these components with some equivalent algorithms.

We tried to use a fixed-parameter algorithm instead of a greedy one for searching of an independent set.  But even the exact solution of this subproblem improved the whole heuristic insignificantly.  Hence, the crucial part of SGA is building of a spanning tree.  We tried three different algorithms, and it appears that the Depth-First search is clearly superior to the other ones in building of the spanning tree for SGA.

Thereby, by application of fixed-parameter algorithms, we managed to check that the existing SGA heuristic is of a high quality and selected the component which required an improvement.  This allowed us to intensify the research in a proper direction which yielded a superior variation of SGA\@.  This variation significantly improves the results of the basic SGA solving most of the instances in our experiments to optimality in a short time.
\end{abstract}

\section{Introduction, terminology and notation}
Large-scale LP models which arise in applications usually have sparse coefficient matrices with special structure. If a special structure can be recognized, it can often be used to considerably speed up the process of solving the LP problem and/or to help in understanding the nature of the LP model. A well-known family of such special structures is networks; a number of
heuristics to extract (reflected) networks in LP problems have been developed and analyzed, see, e.g., \cite{BM93,BF88,Bixby,BW84,GGMM,gulpinar,Hsu} (a formal definition of a reflected network is given below). From the computational point of view, it is worthwhile extracting a reflected network only if the LP problem under consideration contains a relatively large reflected network.

We consider an LP problem in the standard form stated as
\begin{center}
Minimize $\{p^{T}x; \mbox{ subject to } Ax=b,\; x\ge 0\}.$
\end{center}
LP problems have a number of equivalent, in a sense, forms that can be obtained from each other by various operations. Often scaling operations, that is multiplications of rows and columns of the
matrix $A$ of constraints by non-zero constants, are applied, see, e.g., \cite{BM93,BF88,BW84,GGMM}. In the sequel unless stated otherwise, we assume that certain scaling operations on $A$ have been carried out and will not be applied again apart from row reflections defined below.
A matrix $B$ is a {\em network (matrix)} if $B$ is a ($0,\pm 1$)-{\em matrix} (that is, entries of $B$ belong to the set $\{1,0,-1\}$) and every column of $B$ has at most one entry equal to 1 and at
most one entry equal to $-1.$ The operation of {\em reflection} of a row of a matrix $B$ changes the signs of all non-zero entries of this row. A matrix $B$ is a {\em reflected network (matrix)} if
there is a sequence of row reflections that transforms $B$ into a network matrix. The {\em problem of detecting a maximum embedded reflected network} (DMERN) is to find the maximum number of rows that
form a submatrix $B$ of $A$ such that $B$ is a reflected network. This number is denoted by $\nu(A)$. The DMERN problem is known to be NP-hard \cite{Bartholdi}.

G\"{u}lp\i nar et al.\ \cite{gulpinar} showed that the maximum size of an embedded reflected network equals the maximum order of a balanced induced subgraph of a special signed graph associated with matrix $A$ (for details, see Section \ref{secENSG}). This result led G\"{u}lp\i nar et al.\ \cite{gulpinar} to a heuristic named SGA for detection of reflected networks. Computational experiments in \cite{gulpinar} with SGA and three other heuristics demonstrated that SGA and another heuristic, RSD, were of very similar quality and clearly outperformed the two other heuristics in this respect. However, SGA was about 20 times faster, on average, than RSD\@.  Moreover, SGA has an important theoretical property that RSD does not have: SGA always solves the DMERN problem to optimality when the whole matrix $A$ is a reflected network \cite{gulpinar}. Since SGA appeared to be the best choice for a heuristic for detection of reflected networks, Gutin and Zverovitch \cite{GZ} investigated `repetition' versions of SGA and found out that three times repetition of SGA (SGA3) gives about 1\% improvement, while 80 times repetition of SGA (SGA80) leads to 2\% improvement.

In this paper we propose a more refined analysis of the SGA heuristic.  By using a fixed-parameter algorithm, we managed to find the optimal solutions for majority of the instances in the testbed.  This helped us to understand that SGA, in fact, obtains solutions of very high quality and sometimes even solves the instances to optimality.  This means that even a small improvement of SGA quality is a significant achievement.  Having this result, we tried to improve SGA.

SGA contains two components which leave some space for improvement.  One is an independent set searching algorithm.  In the original version, a greedy algorithm was used for this purpose.  We replaced it with a fixed-parameter algorithm.  Here we used the well-known fact that the complement of an independent set in a graph is a vertex cover.  However, our experiments have shown that even the greedy algorithm usually finds the optimal or almost optimal solutions and, thus, this modification of SGA appears to be of little practical interest.  This demonstrated that the independent set extracting heuristic need not be replaced by a more powerful heuristic or exact algorithm.

Hence, the crucial component in SGA is the second component, i.e., building of a spanning tree.  In the original version of SGA (\cite{gulpinar,GZ}) a random procedure was used for this purpose.  We tried to replace it with Breadth-First Search (BFS) and Depth-First search (DFS) algorithms.  The experiments show that the final solution quality significantly depends on spanning tree and that using DFS is clearly superior to both BFS and Random Search (RS).

Moreover, we observe that the choice of the algorithm for spanning tree computation does not essentially influence the runtime of the heuristic under consideration.  Thus we propose three new heuristics SGA(DFS), SGA3(DFS), and SGA80(DFS) that take about the same time as their respective counterparts SGA, SGA3, and SGA80 but are more precise.

In order to evaluate the quality of the considered heuristics we compare their outputs with \emph{optimal} solutions of the considered instances.  To solve the instances to optimality, we design a \emph{fixed-parameter algorithm} for the DMERN problem.  Such algorithm is polynomial in the input size but exponential in terms of an additional parameter associated with the given problem.  Problems that can be solved by fixed-parameter algorithms are called fixed-parameter tractable (FPT)\@.  A fixed-parameter algorithm is usually applied when the parameter is small. In this situation the exponential component of the runtime becomes a relatively small multiplicative or additive constant, that is the problem is solved by a polynomial (usually even a low polynomial) algorithm.  Nevertheless, even if the parameter is small, the researchers often prefer to use imprecise heuristic methods simply because they are faster. We argue that in this situation a fixed-parameter algorithm may be still of a considerable use because it can help to evaluate the quality of the considered heuristics. In particular, in our experiments the use of a fixed-parameter algorithm allowed to observe that the heuristic SGA80(DFS) almost always returns an optimal solution.  That is, although the heuristic is the slowest among the considered ones, this is still not the reason to discard it: the heuristic can be the best choice when the quality is particularly crucial.  Thus we represent a novel way of application of fixed-parameter algorithms which, we believe, would be a significant contribution to the applied research related to fixed-parameter computation.

To design a fixed-parameter algorithm for the DMERN problem we in fact design a fixed-parameter algorithm for the maximum balanced subgraph problem using the equivalence result from \cite{gulpinar}. The fixed-parameter algorithm for the latter problem is designed by showing its equivalence to the bipartization problem and then using a fixed-parameter algorithm for the bipartization problem first proposed in   \cite{RSV04} and then refined in \cite{Huff2005}.

The rest of the paper is organized as follows. In Section \ref{secENSG} we introduce necessary notation, Section \ref{secSGA} presents the SGA heuristic and its
variants, and Section \ref{secFPA} introduces the fixed-parameter algorithms.
Section \ref{secEA} describes a fixed-parameter algorithm for the maximum balanced subgraph problem. In Section \ref{secEE} we report empirical results and
analyze them.  Concluding remarks are made in Section \ref{sec:C}.

\section{Embedded networks and signed graphs}\label{secENSG}

In this section, we assume, for simplicity, that $A$ is a ($0,\pm 1$)-matrix itself (since all rows containing entries not from the set $\{-1,0,+1\}$ cannot be part of a reflected network).
Here we allow graphs to have parallel edges, but no loops. A graph $G=(V,E)$ along with a function $s:\ E\dom \{-,+\}$ is called a {\em signed graph}. Signed graphs have been studied by many researchers, see, e.g., \cite{Ha78,HK80,HT56,Za82}.

We assume that signed graphs have no parallel edges of the same sign, but may have parallel edges of opposite signs. An edge is {\em positive} ({\em negative}) if it is assigned plus (minus).
For a ($0,\pm 1$)-matrix  $A=[a_{ik}]$ with $n$ rows, we construct a signed graph $G(A)$ as follows: the vertex set of $G(A)$ is $\{1,2,\ldots ,n\}$;  $G(A)$ has a positive (negative) edge $ij$ if and only if $a_{ik}=-a_{jk}\neq 0$ ($a_{ik}=a_{jk}\neq 0$) for some $k$. Let $G=(V,E,s)$
be a signed graph. For a non-empty subset $W$ of $V$, the {\em $W$-switch} of $G$ is the signed graph $G^W$ obtained from $G$ by
changing the signs of the edges between $W$ and $V(G)\setminus W$. A signed graph $G=(V,E,s)$ is {\em balanced} if there exists a subset $W$ of $V$ ($W$ may coincide with $V$) such that $G^W$ has no negative edges. Let $\eta(G)$ be the largest order of a balanced
induced subgraph of $G$.

The following important result was proved in \cite{gulpinar}. This result allows us to search for a largest balanced
induced subgraph of $G(A)$ instead of a largest reflected network in $A$.

\begin{theorem}\cite{gulpinar}
\label{mawglem} Let $A$ be a $(0,\pm 1)$-matrix. A set $R$ of rows in $A$ forms a reflected network if and only if the vertices of $G(A)$ corresponding to $R$ induce a balanced subgraph of $G(A)$.
In particular, $\nu(A)=\eta(G(A)).$
\end{theorem}

\section{SGA and its Variations}\label{secSGA}

The heuristic SGA introduced in \cite{gulpinar} is based on the following:

\begin{lemma}\cite{gulpinar}
\label{l} Every signed tree $T$ is a balanced graph.
\end{lemma}
\begin{proof} We prove the lemma by induction on the number of edges in $T$.
The lemma is true when the number of edges is one. Let $x$ be a
vertex of $T$ of degree one. By the induction hypothesis, there
is a set $W\subseteq V(T)-x$ such that $(T-x)^W$ has no negative edges. In $T^W$
the edge $e$ incident to $x$ is positive or negative. In the
first case, let $W'=W$ and the second case, let $W'=W\cup \{x\}$.
Then, $T^{W'}$ has no negative edges. \end{proof}

\begin{center}
\fbox{
\parbox{0.9\linewidth}{

{\bf Heuristic SGA:}

{\it Step 1:} Construct signed graph $G=G(A)=(V,E,s)$.

{\it Step 2:} Find a spanning forest $T$ in $G$.

{\it Step 3:} Using a recursive algorithm based on the proof of
Lemma \ref{l}, compute \quad \quad \quad \quad $W\subseteq V$
such that $T^W$ has no negative edges.

{\it Step 4:} Let  $N$ be the subgraph of $G^W$ induced by the negative edges. Apply the following greedy-degree algorithm \cite{Pa92} to find a maximal independent set
$I$ in $N$: starting from empty $I$, append to $I$ a vertex of $N$ of minimum
degree, delete this vertex together with its neighbors from $N$,
and repeat the above procedure till $N$ has no vertex.

{\it Step 5:} Output $I.$
}}
\end{center}

\2

For a graph $H$ with connectivity components $H_1,\ldots ,H_p$, 
a {\em spanning forest} is the union of spanning trees $T_1,\ldots ,T_p$ of $H_1,\ldots ,H_p$, respectively.
The second step of the algorithm (i.e., finding of spanning forest) can be implemented in a number of different ways. We tried the following ones:
\begin{enumerate}
\item \textit{Random Search (RS)} starts from marking some vertex.  On every iteration, it finds some edge (positive or negative; double edges are not considered) between a marked vertex $v$ and an unmarked vertex $u$.  It marks $u$ and includes or excludes it from the set $W$ according to the sign of the edge $uv$.  If no edges between marked and unmarked vertices are found, the algorithm marks some random vertex again.
    This was the method of computing a spanning forest used for the experiments reported in \cite{gulpinar}.

\item \textit{BFS} is a well known algorithm for constructing of spanning forests.  It maintains a FIFO queue.  On every iteration, it takes a vertex from this queue, marks all its unmarked neighbors into the queue and adds these neighbors to the queue.  If the queue is empty, some unmarked vertex is marked and added into the queue.  In our implementation, we take a vertex of maximum degree in this case.

\item \textit{DFS} is another well known algorithm for constructing of spanning forests.  It starts from some vertex and then calls itself recursively for every of its neighbors which still are not included in the spanning forest.
\end{enumerate}

\begin{figure}
\begin{center}

 \setlength{\unitlength}{0.00035in}
\begingroup\makeatletter\ifx\SetFigFont\undefined%
\gdef\SetFigFont#1#2#3#4#5{%
  \reset@font\fontsize{#1}{#2pt}%
  \fontfamily{#3}\fontseries{#4}\fontshape{#5}%
  \selectfont}%
\fi\endgroup%
{\renewcommand{\dashlinestretch}{30}
\begin{picture}(12242,2175)(0,-10)
\drawline(150.000,675.000)(219.982,640.385)(291.952,610.119)
    (365.639,584.314)(440.767,563.068)(517.055,546.460)
    (594.216,534.553)(671.962,527.390)(750.000,525.000)
    (828.038,527.390)(905.784,534.553)(982.945,546.460)
    (1059.233,563.068)(1134.361,584.314)(1208.048,610.119)
    (1280.018,640.385)(1350.000,675.000)
\drawline(9450.000,675.000)(9380.018,709.615)(9308.048,739.881)
    (9234.361,765.686)(9159.233,786.932)(9082.945,803.540)
    (9005.784,815.447)(8928.038,822.610)(8850.000,825.000)
    (8771.962,822.610)(8694.216,815.447)(8617.055,803.540)
    (8540.767,786.932)(8465.639,765.686)(8391.952,739.881)
    (8319.982,709.615)(8250.000,675.000)
\drawline(8250.000,675.000)(8319.982,640.385)(8391.952,610.119)
    (8465.639,584.314)(8540.767,563.068)(8617.055,546.460)
    (8694.216,534.553)(8771.962,527.390)(8850.000,525.000)
    (8928.038,527.390)(9005.784,534.553)(9082.945,546.460)
    (9159.233,563.068)(9234.361,584.314)(9308.048,610.119)
    (9380.018,640.385)(9450.000,675.000)
\put(150,1875){\circle{150}} \put(1350,1875){\circle{150}}
\put(150,675){\circle{150}} \put(1350,675){\circle{150}}
\put(4050,675){\circle{150}} \put(5550,1875){\circle{150}}
\put(5550,675){\circle{150}} \put(6750,1875){\circle{150}}
\put(6750,675){\circle{150}} \put(8250,1875){\circle{150}}
\put(8250,675){\circle{150}} \put(9450,1875){\circle{150}}
\put(9450,675){\circle{150}} \put(10950,1875){\circle{150}}
\put(12150,1875){\circle{150}} \put(10950,675){\circle{150}}
\put(12150,675){\circle{150}} \put(2850,675){\circle{150}}
\put(2850,1875){\circle{150}} \put(4050,1875){\circle{150}}
\drawline(150,1875)(1350,1875) \drawline(1350,1875)(1350,675)
\drawline(150,1875)(150,675) \drawline(150,1875)(1350,675)
\drawline(2850,675)(2850,1875) \drawline(2850,1875)(4050,1875)
\drawline(4050,1875)(4050,675) \drawline(5550,675)(5550,1875)
\drawline(5550,1875)(6750,1875) \drawline(6750,1875)(6750,675)
\drawline(8250,675)(8250,1875) \drawline(8250,1875)(9450,1875)
\drawline(9450,1875)(9450,675) \drawline(9450,675)(8250,1875)
\drawline(10950,675)(12150,675) \drawline(12150,675)(10950,1875)
\put(750,375){\makebox(0,0)[lb]{\smash{{{$-$}}}}}
\drawline(1350.000,675.000)(1280.018,709.615)(1208.048,739.881)
    (1134.361,765.686)(1059.233,786.932)(982.945,803.540)
    (905.784,815.447)(828.038,822.610)(750.000,825.000)
    (671.962,822.610)(594.216,815.447)(517.055,803.540)
    (440.767,786.932)(365.639,765.686)(291.952,739.881)
    (219.982,709.615)(150.000,675.000)
\put(750,900){\makebox(0,0)[lb]{\smash{{{$+$}}}}}
\put(11550,0){\makebox(0,0)[lb]{\smash{{{$M$}}}}}
\put(-80,1275){\makebox(0,0)[lb]{\smash{{{$+$}}}}}
\put(750,1950){\makebox(0,0)[lb]{\smash{{{$-$}}}}}
\put(1425,1275){\makebox(0,0)[lb]{\smash{{{$-$}}}}}
\put(750,1350){\makebox(0,0)[lb]{\smash{{{$-$}}}}}
\put(2620,1275){\makebox(0,0)[lb]{\smash{{{$+$}}}}}
\put(3450,1950){\makebox(0,0)[lb]{\smash{{{$-$}}}}}
\put(4100,1275){\makebox(0,0)[lb]{\smash{{{$-$}}}}}
\put(5300,1275){\makebox(0,0)[lb]{\smash{{{$+$}}}}}
\put(6150,1950){\makebox(0,0)[lb]{\smash{{{$+$}}}}}
\put(6825,1275){\makebox(0,0)[lb]{\smash{{{$+$}}}}}
\put(8850,1950){\makebox(0,0)[lb]{\smash{{{$+$}}}}}
\put(8000,1275){\makebox(0,0)[lb]{\smash{{{$+$}}}}}
\put(8925,1275){\makebox(0,0)[lb]{\smash{{{$-$}}}}}
\put(9525,1275){\makebox(0,0)[lb]{\smash{{{$+$}}}}}
\put(8850,375){\makebox(0,0)[lb]{\smash{{{$-$}}}}}
\put(8850,900){\makebox(0,0)[lb]{\smash{{{$+$}}}}}
\put(150,2025){\makebox(0,0)[lb]{\smash{{{$1$}}}}}
\put(1350,2025){\makebox(0,0)[lb]{\smash{{{$2$}}}}}
\put(150,375){\makebox(0,0)[lb]{\smash{{{$3$}}}}}
\put(1350,375){\makebox(0,0)[lb]{\smash{{{$4$}}}}}
\put(2850,2025){\makebox(0,0)[lb]{\smash{{{$1$}}}}}
\put(4050,2025){\makebox(0,0)[lb]{\smash{{{$2$}}}}}
\put(2850,375){\makebox(0,0)[lb]{\smash{{{$3$}}}}}
\put(4050,375){\makebox(0,0)[lb]{\smash{{{$4$}}}}}
\put(5550,2025){\makebox(0,0)[lb]{\smash{{{$1$}}}}}
\put(6750,2025){\makebox(0,0)[lb]{\smash{{{$2$}}}}}
\put(5550,375){\makebox(0,0)[lb]{\smash{{{$3$}}}}}
\put(6750,375){\makebox(0,0)[lb]{\smash{{{$4$}}}}}
\put(8250,2025){\makebox(0,0)[lb]{\smash{{{$1$}}}}}
\put(9450,2025){\makebox(0,0)[lb]{\smash{{{$2$}}}}}
\put(8250,375){\makebox(0,0)[lb]{\smash{{{$3$}}}}}
\put(9450,375){\makebox(0,0)[lb]{\smash{{{$4$}}}}}
\put(10950,2025){\makebox(0,0)[lb]{\smash{{{$1$}}}}}
\put(12150,2025){\makebox(0,0)[lb]{\smash{{{$2$}}}}}
\put(10950,375){\makebox(0,0)[lb]{\smash{{{$3$}}}}}
\put(12150,375){\makebox(0,0)[lb]{\smash{{{$4$}}}}}
\put(750,0){\makebox(0,0)[lb]{\smash{{{$G$}}}}}
\put(3450,0){\makebox(0,0)[lb]{\smash{{{$T$}}}}}
\put(6150,0){\makebox(0,0)[lb]{\smash{{{$T^{\{2\}}$}}}}}
\put(8850,0){\makebox(0,0)[lb]{\smash{{{$G^{\{2\}}$}}}}}
\end{picture}
}

\caption{Illustration for SGA; $M$ is the subgraph $G^{\{2\}}$ induced by the negative edges of $G^{\{2\}}$.}\label{sgfig}
\end{center}
\end{figure}
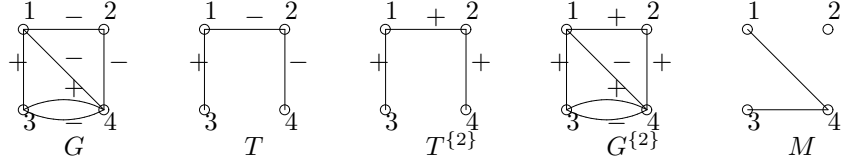

\begin{proposition}\cite{gulpinar}
If $G$ is balanced, then $I=V.$
\end{proposition}
\begin{proof}
It is well-known (see, e.g., Theorem 2.8 in \cite{gulpinar}) that a signed graph is balanced if and only if it does not contain cycles
with odd number of negative edges. Let $T$ be a a spanning forest in $G$. Since $T^W$ has no negative edges, $G^W$ cannot have negative edges. Indeed, if $xy$ was a negative edge in $G^W,$ it would be the unique negative edge in a cycle formed by $xy$ and the $(x,y)$-path of $T^W,$ a contradiction.
\end{proof}

Gutin and Zverovitch \cite{GZ} investigated a repetition version of SGA where Steps 2-4 were repeated several time (each time the vertices of $G$ were pseudo-randomly permuted and a new spanning forest of $G$ was built). They found out that three times repetition of SGA gives about 1\% improvement, while 80 times repetition of SGA leads to 2\% improvement, on average.
In our experiments we used a larger test bed and better scaling procedure than in \cite{GZ} and, thus, we run SGA and its 3 and 80 times repetitions on the new set of instances of the DMERN problem (see Section \ref{secEE}). We will denote these repetition versions of SGA by SGA3 and SGA80, respectively.

In Section \ref{secEE} we also report results on another modification of SGA, SGA+VC, where we replace Step 4 with finding a vertex cover $C$ of $G^W$ and setting $I=V(G^W)\setminus C.$
Since the vertex cover problem is well studied in the area of parameterized
complexity \cite{AFLS,chen2005,niedermeier2006}, to find $C$ we can use a fixed-parameter algorithm for the problem.

\section{Fixed-Parameter Algorithmics}\label{secFPA}

We recall some most basic notions of fixed-parameter algorithmics (FPA) here, for a
more in-depth treatment of the topic we refer the reader to the monographs
\cite{downey1999,flum2006,niedermeier2006}.

FPA is a relatively new approach for dealing with intractable
computational problems. In the framework of FPA we introduce a
parameter $k$, which is often a positive integer (but may be a
vector, graph, or any other object for some problems) such that the
problem at hand can be solved in time $O(f(k)n^c)$, where $n$ is the
size of the problem instance, $c$ is a constant not dependent on $n$
or $k$, and $f(k)$ is an arbitrary computable function not dependent
on $n$.  The ultimate goal is to obtain $f(k)$ and $c$ such that for
small or even moderate values of $k$ the problem under consideration
can be completely solved in a reasonable amount of time.

As an example, consider the {\em Vertex Cover problem (VC)}: given
an undirected graph $G$ (with $n$ vertices and $m$ edges), find a
minimum number of vertices such that every edge is incident to at
least one of these vertices. In the (naturally) parameterized
version of VC, $k$-VC, given a graph $G$, we are to check whether
$G$ has a vertex cover with at most $k$ vertices. $k$-VC admits an
algorithm of running time $O(1.2738^k+kn)$ obtained in \cite{chen2005} that allows us to solve VC with $k$ up to several
hundreds. Without using FPA, we would be likely to end up with the
obvious algorithm of complexity $O(mn^k)$. The last algorithm is far
too slow even for small values of $k$ such as $k=10$.

Parameterized problems that admit algorithms of complexity
$O(f(k)n^c)$ (we refer to such algorithms as {\em fixed-parameter})
are called {\em fixed-parameter tractable (FPT)}\@.  Notice that not every parameterized problem is FPT, but there are
many problems that are FPT\@ \cite{downey1999,flum2006,niedermeier2006}.

\section{Minimum Balanced Deletion problem}\label{secEA}

By our discussions above, we are interested in the following parameterized problem.

\begin{center}
\fbox{
\parbox{0.9\linewidth}{
  \noindent{\bfseries  The minimum balanced deletion problem (MBD)}\\
  \emph{Input:} A signed graph $G=(V,E,s)$, an integer $k$.\\
  \emph{Parameter:} $k$.\\
  \emph{Output:} A set of at most $k$ vertices whose removal makes $G$ balanced or
  'NO' if no such set exists.
}}
\end{center}

We show that the MBD problem is FPT by transforming it into the Bipartization problem
defined as follows.

\begin{center}
\fbox{
\parbox{0.9\linewidth}{
  \noindent{\bfseries  The Bipartization problem}\\
  \emph{Input:} A graph $G$, an integer $k$\\
  \emph{Parameter:} $k$\\
  \emph{Output:} A set of at most $k$ vertices whose removal makes $G$ bipartite or
  'NO' if no such set exists.
}}
\end{center}

The transformation is described in the following theorem.

\begin{theorem}
The MBD problem is FPT and can be solved in time $O^*(3^k)$.
\end{theorem}
\begin{proof}
It is well-known (see, e.g., Theorem 2.8 in \cite{gulpinar}) that a signed graph is balanced if and only if it does not contain cycles
involving odd number of negative edges. Hence, the MBD problem in fact asks
for at most $k$ vertices whose removal breaks all cycles containing
an odd number of negative edges.

Let $G'$ be the (unsigned) graph obtained
from $G$ by \emph{subdividing} each positive edge. In
other words, for each positive edge $\{u,v\}$, we introduce a new vertex
$w$ and replace $\{u,v\}$ by $\{u,w\}$ and $\{w,v\}$. We claim that $G$ has
a set of at most $k$ vertices breaking all cycles with an odd number of negative
edges if and only if $G'$ can be made bipartite by removal of at most $k$ vertices.

Assume the former and let $K$ be a set of at most $k$ vertices whose removal
breaks all cycles with an odd number of negative edges. It follows that
$G' - K$ is bipartite. Indeed, each cycle $C'$ of $G' - K$
can be obtained from a cycle $C$ of $G - K$ by subdivision of its
positive edges. Hence, $C'$ can be of an odd length only if $C$ has an odd number
of negative edges which is impossible according to our assumption about $K$.

Conversely, let $K$ be a set of at most $k$ vertices such that $G' - K$
is bipartite. We may safely assume that $K$ does not contain the new vertices
subdividing positive edges: otherwise each such vertex can be replaced by one
of its neighbors. Thus, $K \subseteq V(G)$. Observe that $G - K$
does not have cycles with odd number of negative edges. Indeed, by subdividing
positive edges, any such cycle translates into an odd cycle of $G' - K$
in contradiction to our assumption about $K$.

It follows from the above argumentation that the MBD problem can be solved as
follows. Transform $G$ into $G'$ and run on $G'$ the $O^*(3^k)$ algorithm
solving the bipartization problem \cite{Huff2005}. If the algorithm returns 'NO'
then return 'NO'. Otherwise, replace each subdividing vertex by one of its
neighbors and return the resulting set of vertices. Clearly, the complexity
of the resulting algorithm is $O^*(3^k)$.  \end{proof}

\smallskip

{\bf Remarks.} The usual trick to avoid the subdivided vertices to be selected to the resulting solution would be to make $k+1$ copies of each vertex. However, such approach would increase the runtime of the resulting implementation and hence we have used a slightly more sophisticated method. Unfortunately, it is not known yet whether the Bipartization problem has a polynomial-size problem kernel \cite{Huff2005}. Thus, it is not known yet whether the MBD problem has a polynomial-size problem kernel. (If one was known, we could try to use it to speed up our fixed-parameter algorithm.)

Note that a version of the MBD problem, where edge-deletions rather than vertex-deletions are
used was considered in \cite{BHN077,DESZ06}.

\section{Experimental Evaluation}\label{secEE}

In this section we provide and discuss our experiment results for the heuristics SGA, SGA3, SGA80, SGA+VC descried in Section \ref{secSGA} and the exact algorithm given in Section \ref{secEA}.
Note that in our experiments we use a larger test bed and better scaling procedure than in \cite{GZ}.

\subsection{Scaling Procedure}

Recall that we consider an LP problem in the standard form stated as
\begin{center}
Minimize $\{p^{T}x; \mbox{ subject to } Ax=b,\; x\ge 0\}.$
\end{center}
In Section \ref{secENSG}, to simplify our notation we assumed that $A$ is a ($0,\pm 1$)-matrix. However, in general, in real LP problems  $A$ is not a ($0,\pm 1$)-matrix.
Therefore, in reality, the first phase in solving the DMERN problem is applying a scaling procedure whose aim
is to increase the number of $(0, \pm{}1)$-rows by scaling rows and columns. Here we describe a scaling procedure that we have used. Our computational experiments indicate that this scaling is often better than the scaling procedures we found in the literature. Let us describe our scaling procedure. Let $A=[a_{ij}]_{n \times m}.$

First we apply simple row scaling, i.e., scale all the rows which contain only zeros and $\pm{} x$, where $x > 0$ is some constant: for every $i \in \{ 1, 2, \ldots, n \}$ set $a_{ij} = a_{ij} / x$ for $j = 1, 2, \ldots, m$ if $a_{ij} \in \{0, -x, +x\}$ for every $j \in \{1, 2, \ldots, m\}$.

Then we apply a more sophisticated procedure.  Let $[r_i]_n$ be an array of boolean values, where $r_i$ indicates whether the $i$th row is a $(0, \pm{}1)$-row.  Let $[b_j]_m$ be an array of boolean values, where $b_j$ indicates whether the $j$th column is bounded, i.e., whether it has at least one nonzero value in a $(0, \pm{}1)$-row: for some $j \in \{ 1, 2, \ldots, m \}$ the value $b_j = true$ if and only if there exists some $i$ such that $r_i = true$ and $a_{ij} \neq 0$.

Next we do the following for every non $(0, \pm{} 1)$-row (note that at this stage any non $(0, \pm{}1)$-row contains at least two nonzero elements).  Let $J$ be the set of indices of bounded columns with nonzero elements in the current row $c$: $J = \{ j: a_{cj} \neq 0 \mbox{ and } b_j = true \}$.  If $J = \emptyset$, i.e., all the columns corresponding to nonzero elements in the current row are unbounded, then we simply scale every of these columns: $a_{ij} = a_{ij} / a_{cj}$ for every $i = 1, 2, \ldots, n$ and for every $j$ such that $a_{cj} \neq 0$.  If $J \neq \emptyset$ and $a_{cj} \in \{+x, -x\}$ for every $j \in J$, where $x$ is some constant, then we scale accordingly the current row ($a_{cj} = a_{cj} / x$ for every $j \in \{ 1, 2, \ldots, m \}$) and scale the unbounded columns: $a_{ij} = a_{ij} / a_{cj}$ for every $j \notin J$ if $a_{cj} \neq 0$.  Otherwise we do nothing for the current row.

Every time when we scale rows or columns we update the arrays $r$ and $b$.

Since the matrices processed by this heuristic are usually sparse, we use a special data structure to store them.  In particular, we store only nonzero elements providing the row and column indices for each of them.  We also store a list of references to the corresponding nonzero elements for every row and for every column of the matrix.

\subsection{Computational Experience}

The computational results for all heuristics apart from SGA+VC as well as for the exact algorithm are provided in Table~\ref{tab:experiments}.  As a test bed we use all the instances provided in Netlib (\url{http://netlib.org/lp/data/}). All algorithms were implemented in C++ and the evaluation platform is based on an AMD Athlon 64 X2 3.0~GHz processor. For the exact algorithm we used a code of H{\"u}ffner \url{http://theinf1.informatik.uni-jena.de/~hueffner/}. In SGA+VC we used a vertex cover code based on \cite{ALSS}.
Since the exact algorithm can potentially take a very long time, we introduced a timeout: if after one hour of work
the algorithm is unable to compute the output, it terminates.

Denote by $n$ the number of $(0,\pm 1)$-rows in the instance, i.e., the number of vertices in the corresponding signed graph $G$.  The column $k$ reports the difference between $n$ and the number of vertices in a maximum induced balanced subgraph of $G$ found by the exact algorithm. The rows where $k$ is not given correspond to the instances where the algorithm terminated after one hour without computing the output. The nine columns following column $k$ report the same differences found by heuristics SGA, SGA3, and SGA80\@.  The first three columns are related to SGA\@. The columns RS, BFS, and DFS are related to the way of computing the spanning tree (see Section 3 for the detailed explanation).  The next three columns are related to SGA3 and the last three columns are related to SGA80 with the analogous meaning of particular columns. The column $t$ reports the runtime taken by the exact algorithm.  The time provided in the column `$t_1$' is the average time required for SGA(RS), SGA(BFS) and SGA(DFS) \footnote{the acronym in the parenthesis correspond to the algorithm of computing the spanning tree.} to proceed once.  Our experiments show that for SGA3 this time is approximately 3 times larger and for SGA80 80 times larger.  It also appears that there is no significant difference between running times of SGA based on RS, BFS or DFS.

The 'Average' row shows the average value of the respective columns.
The `Avg. diff.'\ row shows how far, on average, is a particular modification of SGA from the optimal solutions.  `\# exact sol.'\ shows the number of instances for which a particular modification of SGA obtained an optimal solution.  Both `Avg. diff.'\ and `\# exact sol.'\ consider only the instances with the knows optimal solutions.

{\footnotesize
\begin{longtable}{
@{}
l @{~} c @{~} r
c *{3}{@{~} r}
c *{3}{@{~} r}
c *{3}{@{~} r}
c *{2}{@{~~} r}
@{}}
\\
\caption{Experiment results.}\\
\label{tab:experiments} \\

\toprule
 & ~ &  & ~ & \multicolumn{3}{c}{SGA} & ~ & \multicolumn{3}{c}{SGA3} & ~ & \multicolumn{3}{c}{SGA80} & ~ & \multicolumn{2}{c}{Time, s} \\
\cmidrule{5-7} \cmidrule{9-11} \cmidrule{13-15} \cmidrule{17-18}
Instance	&&	$k$	&&	RS	&	BFS	&	DFS	&&	RS	&	BFS	&	DFS	&&	RS	&	BFS	&	DFS	&&	$t$	&	$t_1$	\\
\cmidrule(){1-18}

\endfirsthead

\toprule
 & ~ &  & ~ & \multicolumn{3}{c}{SGA} & ~ & \multicolumn{3}{c}{SGA3} & ~ & \multicolumn{3}{c}{SGA80} & ~ & \multicolumn{2}{c}{Time, s} \\
\cmidrule{5-7} \cmidrule{9-11} \cmidrule{13-15} \cmidrule{17-18}
Instance	&&	$k$	&&	RS	&	BFS	&	DFS	&&	RS	&	BFS	&	DFS	&&	RS	&	BFS	&	DFS	&&	$t$	&	$t_1$	\\
\cmidrule(){1-18}

\endhead
\bottomrule
\endfoot

25FV47               	&&	15	&&	25	&	38	&	21	&&	25	&	38	&	18	&&	22	&	38	&	16	&&	4.40	&	0.004	\\
80BAU3B              	&&	 --- 	&&	42	&	40	&	40	&&	40	&	40	&	40	&&	40	&	40	&	40	&&	 $>1$h 	&	0.106	\\
ADLITTLE             	&&	1	&&	1	&	1	&	1	&&	1	&	1	&	1	&&	1	&	1	&	1	&&	0.02	&	0.000	\\
AFIRO                	&&	0	&&	0	&	0	&	0	&&	0	&	0	&	0	&&	0	&	0	&	0	&&	0.00	&	0.000	\\
AGG                  	&&	 --- 	&&	107	&	108	&	104	&&	106	&	108	&	104	&&	104	&	108	&	103	&&	 $>1$h 	&	0.001	\\
AGG2                 	&&	 --- 	&&	85	&	91	&	83	&&	85	&	91	&	83	&&	83	&	91	&	83	&&	 $>1$h 	&	0.001	\\
AGG3                 	&&	 --- 	&&	85	&	91	&	83	&&	85	&	91	&	83	&&	83	&	91	&	83	&&	 $>1$h 	&	0.001	\\
BANDM                	&&	23	&&	24	&	24	&	24	&&	23	&	24	&	23	&&	23	&	24	&	23	&&	1493.12	&	0.001	\\
BEACONFD             	&&	3	&&	3	&	3	&	3	&&	3	&	3	&	3	&&	3	&	3	&	3	&&	0.00	&	0.000	\\
BLEND                	&&	1	&&	1	&	1	&	1	&&	1	&	1	&	1	&&	1	&	1	&	1	&&	0.00	&	0.000	\\
BNL1                 	&&	14	&&	17	&	19	&	15	&&	17	&	18	&	14	&&	14	&	18	&	14	&&	1.83	&	0.003	\\
BNL2                 	&&	 --- 	&&	127	&	109	&	103	&&	110	&	107	&	101	&&	99	&	105	&	86	&&	 $>1$h 	&	0.061	\\
BOEING1              	&&	 --- 	&&	49	&	61	&	42	&&	49	&	61	&	42	&&	48	&	60	&	42	&&	 $>1$h 	&	0.001	\\
BOEING2              	&&	15	&&	17	&	17	&	16	&&	17	&	17	&	16	&&	15	&	17	&	15	&&	0.05	&	0.000	\\
BORE3D               	&&	12	&&	14	&	15	&	13	&&	14	&	15	&	12	&&	12	&	15	&	12	&&	0.14	&	0.001	\\
BRANDY               	&&	6	&&	7	&	8	&	6	&&	6	&	8	&	6	&&	6	&	8	&	6	&&	0.00	&	0.000	\\
CAPRI                	&&	 --- 	&&	40	&	43	&	37	&&	37	&	40	&	34	&&	34	&	40	&	33	&&	 $>1$h 	&	0.001	\\
CYCLE                	&&	 --- 	&&	34	&	36	&	34	&&	34	&	36	&	34	&&	34	&	36	&	34	&&	 $>1$h 	&	0.020	\\
CZPROB               	&&	1	&&	1	&	1	&	1	&&	1	&	1	&	1	&&	1	&	1	&	1	&&	0.27	&	0.018	\\
D2Q06C               	&&	 --- 	&&	67	&	97	&	67	&&	67	&	96	&	67	&&	67	&	94	&	65	&&	 $>1$h 	&	0.042	\\
D6CUBE               	&&	 --- 	&&	61	&	50	&	50	&&	52	&	50	&	50	&&	46	&	47	&	50	&&	 $>1$h 	&	0.003	\\
DEGEN2               	&&	 --- 	&&	234	&	237	&	219	&&	234	&	237	&	219	&&	226	&	234	&	218	&&	 $>1$h 	&	0.011	\\
DEGEN3               	&&	 --- 	&&	822	&	819	&	769	&&	822	&	819	&	769	&&	815	&	819	&	769	&&	 $>1$h 	&	0.197	\\
DFL001               	&&	 --- 	&&	2818	&	2997	&	2603	&&	2818	&	2956	&	2603	&&	2802	&	2903	&	2585	&&	 $>1$h 	&	1.871	\\
E226                 	&&	15	&&	18	&	19	&	15	&&	17	&	19	&	15	&&	16	&	19	&	15	&&	1.09	&	0.001	\\
ETAMACRO             	&&	12	&&	20	&	31	&	19	&&	20	&	31	&	19	&&	20	&	26	&	16	&&	0.47	&	0.001	\\
FFFFF800             	&&	 --- 	&&	50	&	69	&	46	&&	50	&	69	&	46	&&	41	&	65	&	39	&&	 $>1$h 	&	0.002	\\
FINNIS               	&&	 --- 	&&	121	&	127	&	120	&&	121	&	127	&	120	&&	119	&	127	&	119	&&	 $>1$h 	&	0.003	\\
FIT1D                	&&	6	&&	6	&	7	&	6	&&	6	&	7	&	6	&&	6	&	7	&	6	&&	0.00	&	0.000	\\
FIT1P                	&&	0	&&	0	&	0	&	0	&&	0	&	0	&	0	&&	0	&	0	&	0	&&	0.00	&	0.000	\\
FIT2D                	&&	6	&&	7	&	7	&	6	&&	6	&	7	&	6	&&	6	&	7	&	6	&&	0.00	&	0.004	\\
FIT2P                	&&	2	&&	2	&	2	&	2	&&	2	&	2	&	2	&&	2	&	2	&	2	&&	0.00	&	0.007	\\
FORPLAN              	&&	1	&&	1	&	1	&	1	&&	1	&	1	&	1	&&	1	&	1	&	1	&&	0.00	&	0.000	\\
GANGES               	&&	 --- 	&&	83	&	84	&	82	&&	83	&	84	&	82	&&	77	&	84	&	77	&&	 $>1$h 	&	0.021	\\
GFRD-PNC             	&&	 --- 	&&	68	&	68	&	95	&&	68	&	68	&	86	&&	68	&	68	&	80	&&	 $>1$h 	&	0.008	\\
GREENBEA             	&&	 --- 	&&	48	&	60	&	46	&&	48	&	60	&	46	&&	45	&	57	&	41	&&	 $>1$h 	&	0.043	\\
GREENBEB             	&&	 --- 	&&	48	&	60	&	46	&&	48	&	60	&	46	&&	45	&	57	&	41	&&	 $>1$h 	&	0.044	\\
GROW15               	&&	0	&&	0	&	0	&	0	&&	0	&	0	&	0	&&	0	&	0	&	0	&&	0.00	&	0.000	\\
GROW22               	&&	0	&&	0	&	0	&	0	&&	0	&	0	&	0	&&	0	&	0	&	0	&&	0.00	&	0.000	\\
GROW7                	&&	0	&&	0	&	0	&	0	&&	0	&	0	&	0	&&	0	&	0	&	0	&&	0.00	&	0.000	\\
ISRAEL               	&&	8	&&	9	&	10	&	8	&&	8	&	10	&	8	&&	8	&	10	&	8	&&	0.02	&	0.000	\\
KB2                  	&&	1	&&	1	&	3	&	2	&&	1	&	3	&	1	&&	1	&	3	&	1	&&	0.00	&	0.000	\\
LOTFI                	&&	18	&&	24	&	24	&	19	&&	22	&	20	&	19	&&	19	&	20	&	19	&&	11.23	&	0.001	\\
MAROS-R7             	&&	0	&&	0	&	0	&	0	&&	0	&	0	&	0	&&	0	&	0	&	0	&&	0.00	&	0.011	\\
MAROS                	&&	11	&&	17	&	14	&	21	&&	15	&	14	&	18	&&	12	&	14	&	11	&&	0.23	&	0.005	\\
MODSZK1              	&&	 --- 	&&	237	&	267	&	237	&&	237	&	267	&	237	&&	237	&	267	&	235	&&	 $>1$h 	&	0.004	\\
NESM                 	&&	10	&&	13	&	13	&	13	&&	11	&	13	&	11	&&	10	&	13	&	10	&&	0.03	&	0.003	\\
PEROLD               	&&	 --- 	&&	28	&	29	&	25	&&	26	&	29	&	24	&&	24	&	27	&	23	&&	 $>1$h 	&	0.003	\\
PILOT.JA             	&&	16	&&	18	&	19	&	18	&&	17	&	19	&	16	&&	16	&	19	&	16	&&	11.72	&	0.004	\\
PILOT                	&&	 --- 	&&	45	&	45	&	44	&&	42	&	44	&	41	&&	41	&	44	&	41	&&	 $>1$h 	&	0.008	\\
PILOT.WE             	&&	 --- 	&&	34	&	36	&	31	&&	30	&	36	&	30	&&	28	&	36	&	27	&&	 $>1$h 	&	0.003	\\
PILOT4               	&&	3	&&	3	&	3	&	3	&&	3	&	3	&	3	&&	3	&	3	&	3	&&	0.00	&	0.001	\\
PILOT87              	&&	 --- 	&&	77	&	87	&	74	&&	76	&	87	&	71	&&	70	&	87	&	69	&&	 $>1$h 	&	0.015	\\
PILOTNOV             	&&	19	&&	21	&	24	&	21	&&	21	&	24	&	20	&&	19	&	24	&	19	&&	201.29	&	0.005	\\
RECIPE               	&&	0	&&	0	&	0	&	0	&&	0	&	0	&	0	&&	0	&	0	&	0	&&	0.00	&	0.000	\\
SC105                	&&	16	&&	17	&	32	&	22	&&	17	&	31	&	21	&&	17	&	31	&	19	&&	12.56	&	0.000	\\
SC205                	&&	 --- 	&&	36	&	68	&	50	&&	36	&	67	&	47	&&	36	&	67	&	41	&&	 $>1$h 	&	0.000	\\
SC50A                	&&	8	&&	8	&	12	&	8	&&	8	&	11	&	8	&&	8	&	11	&	8	&&	0.02	&	0.000	\\
SC50B                	&&	6	&&	6	&	9	&	8	&&	6	&	9	&	6	&&	6	&	9	&	6	&&	0.02	&	0.000	\\
SCAGR25              	&&	0	&&	0	&	0	&	0	&&	0	&	0	&	0	&&	0	&	0	&	0	&&	0.03	&	0.002	\\
SCAGR7               	&&	0	&&	0	&	0	&	0	&&	0	&	0	&	0	&&	0	&	0	&	0	&&	0.02	&	0.000	\\
SCFXM1               	&&	12	&&	13	&	14	&	13	&&	12	&	14	&	13	&&	12	&	14	&	12	&&	0.30	&	0.001	\\
SCFXM2               	&&	 --- 	&&	26	&	28	&	26	&&	26	&	28	&	26	&&	24	&	28	&	26	&&	 $>1$h 	&	0.003	\\
SCFXM3               	&&	 --- 	&&	39	&	42	&	39	&&	38	&	42	&	39	&&	36	&	42	&	39	&&	 $>1$h 	&	0.006	\\
SCORPION             	&&	1	&&	1	&	1	&	1	&&	1	&	1	&	1	&&	1	&	1	&	1	&&	0.02	&	0.001	\\
SCRS8                	&&	9	&&	9	&	9	&	9	&&	9	&	9	&	9	&&	9	&	9	&	9	&&	0.06	&	0.002	\\
SCSD1                	&&	0	&&	0	&	0	&	0	&&	0	&	0	&	0	&&	0	&	0	&	0	&&	0.03	&	0.000	\\
SCSD6                	&&	0	&&	0	&	0	&	0	&&	0	&	0	&	0	&&	0	&	0	&	0	&&	0.14	&	0.001	\\
SCSD8                	&&	0	&&	0	&	0	&	0	&&	0	&	0	&	0	&&	0	&	0	&	0	&&	0.59	&	0.001	\\
SCTAP1               	&&	0	&&	0	&	0	&	0	&&	0	&	0	&	0	&&	0	&	0	&	0	&&	0.00	&	0.000	\\
SCTAP2               	&&	0	&&	0	&	0	&	0	&&	0	&	0	&	0	&&	0	&	0	&	0	&&	0.00	&	0.006	\\
SCTAP3               	&&	0	&&	0	&	0	&	0	&&	0	&	0	&	0	&&	0	&	0	&	0	&&	0.00	&	0.012	\\
SEBA                 	&&	 --- 	&&	274	&	285	&	267	&&	274	&	285	&	267	&&	269	&	285	&	265	&&	 $>1$h 	&	0.021	\\
SHARE1B              	&&	4	&&	5	&	4	&	4	&&	4	&	4	&	4	&&	4	&	4	&	4	&&	0.00	&	0.000	\\
SHARE2B              	&&	6	&&	6	&	6	&	6	&&	6	&	6	&	6	&&	6	&	6	&	6	&&	0.00	&	0.000	\\
SHELL                	&&	2	&&	2	&	2	&	2	&&	2	&	2	&	2	&&	2	&	2	&	2	&&	1.51	&	0.006	\\
SHIP04L              	&&	 --- 	&&	36	&	36	&	36	&&	36	&	36	&	36	&&	36	&	36	&	36	&&	 $>1$h 	&	0.006	\\
SHIP04S              	&&	 --- 	&&	36	&	36	&	36	&&	36	&	36	&	36	&&	36	&	36	&	36	&&	 $>1$h 	&	0.005	\\
SHIP08L              	&&	 --- 	&&	64	&	64	&	64	&&	64	&	64	&	64	&&	64	&	64	&	64	&&	 $>1$h 	&	0.023	\\
SHIP08S              	&&	 --- 	&&	64	&	64	&	64	&&	64	&	64	&	64	&&	64	&	64	&	64	&&	 $>1$h 	&	0.015	\\
SHIP12L              	&&	 --- 	&&	96	&	96	&	96	&&	96	&	96	&	96	&&	96	&	96	&	96	&&	 $>1$h 	&	0.046	\\
SHIP12S              	&&	 --- 	&&	96	&	96	&	96	&&	96	&	96	&	96	&&	96	&	96	&	96	&&	 $>1$h 	&	0.030	\\
SIERRA               	&&	 --- 	&&	400	&	455	&	420	&&	400	&	455	&	397	&&	387	&	444	&	388	&&	 $>1$h 	&	0.030	\\
STAIR                	&&	8	&&	11	&	12	&	9	&&	9	&	12	&	9	&&	8	&	12	&	8	&&	0.02	&	0.000	\\
STANDATA             	&&	 --- 	&&	53	&	57	&	59	&&	53	&	57	&	59	&&	53	&	57	&	59	&&	 $>1$h 	&	0.002	\\
STANDGUB             	&&	 --- 	&&	53	&	57	&	59	&&	53	&	57	&	59	&&	53	&	57	&	59	&&	 $>1$h 	&	0.002	\\
STANDMPS             	&&	 --- 	&&	54	&	58	&	73	&&	54	&	58	&	73	&&	54	&	58	&	70	&&	 $>1$h 	&	0.004	\\
STOCFOR1             	&&	0	&&	0	&	0	&	0	&&	0	&	0	&	0	&&	0	&	0	&	0	&&	0.00	&	0.000	\\
STOCFOR2             	&&	 --- 	&&	258	&	311	&	202	&&	258	&	307	&	202	&&	243	&	305	&	197	&&	 $>1$h 	&	0.050	\\
TUFF                 	&&	16	&&	26	&	16	&	18	&&	17	&	16	&	16	&&	16	&	16	&	16	&&	0.58	&	0.001	\\
VTP.BASE             	&&	4	&&	6	&	6	&	6	&&	4	&	6	&	4	&&	4	&	6	&	4	&&	0.00	&	0.000	\\
WOOD1P               	&&	0	&&	0	&	0	&	0	&&	0	&	0	&	0	&&	0	&	0	&	0	&&	0.02	&	0.001	\\
WOODW                	&&	0	&&	0	&	0	&	0	&&	0	&	0	&	0	&&	0	&	0	&	0	&&	0.00	&	0.007	\\

\cmidrule(){1-18}

Average	&&		&&	79.3	&	84.85	&	75.57	&&	78.55	&	84.2	&	74.82	&&	76.91	&	83.19	&	73.54	&&	32.26	&	0.030	\\
Avg. diff.	&&		&&	1.28	&	2.15	&	0.93	&&	0.78	&	2.02	&	0.52	&&	0.35	&	1.93	&	0.17	&&		&		\\
\# exact sol. && && 33 & 31 & 36 && 40 & 31 & 44 && 48 & 31 & 50 \\
\end{longtable}
}

Observe that the exact algorithm completed its computations for 54 instances out of the total of 93, and for 52 instances the running time was at most 1 minute.  Note that SGA(DFS) achieved the optimal solution in 36 out of 54 cases, SGA3(DFS) in 44 cases and SGA80(DFS) in 50 cases.

The DFS-based version of SGA clearly outperforms the RS- and BFS-based versions.  This is also true for SGA3 and SGA80.  We are not able to justify this result but it is strongly confirmed by our extensive experimentation.

\bigskip

The results with SGA(DFS)+VC are not provided in Table~\ref{tab:experiments} since SGA(DFS)+VC managed to improve SGA(DFS) only for 4 instances:

\begin{center}
\begin{tabular}{lrrrr}
\toprule
Instance & $k$ & $k_\text{SGA(DFS)+VC}$ & $k_\text{SGA(DFS)}$ & $k_\text{SGA80(DFS)}$ \\
\cmidrule{1-5}
BOEING2 & 15 & 15 & 16 & 15 \\
DEGEN2 & --- & 218 & 219 & 218 \\
DEGEN3 & --- & 764 & 769 & 769 \\
DFL001 & --- & 2599 & 2603 & 2585 \\
\bottomrule
\end{tabular}
\end{center}

SGA(DFS)+VC was not able to proceed in 30 minutes for AGG, AGG2, AGG3, FINNIS, MODSZK1 and SIERRA\@.  Note that $k_\text{SGA(DFS)+VC} < k_\text{SGA80(DFS)}$ for only one instance while $k_\text{SGA(DFS)+VC} > k_\text{SGA80(DFS)}$ for 39 instances and recall that for 6 instances SGA(DFS)+VC did not terminate in the given time.  Since the running time of SGA(DFS)+VC usually exceeds that of SGA80(DFS) and the quality of SGA(DFS)+VC is not much different even from that of SGA(DFS), SGA+VC appears to be of little practical interest.  However, SGA+VC demonstrates that there is no need to replace Step 4 of SGA by a more powerful heuristic or exact algorithm and that in order to improve results one should pay attention to the Step 2 of SGA.  In other words, this negative result pointed out the proper area for further research.  Indeed, replacement of RS algorithm for the spanning tree building with DFS dramatically improved the heuristic quality.

\bigskip

Although SGA80 is slower than SGA and SGA3, it is much more precise and its running time is still reasonable.  It follows that SGA80(DFS) is the best choice with respect to the tradeoff between running time and solution quality.  Observe that in almost all (50 out of 54) the cases feasible for the exact algorithm, SGA80(DFS), being much faster than the exact algorithm, managed to compute an optimal solution!  Note that this conclusion can be made only \emph{given the knowledge} about the optimal solution and without the considered fixed-parameter algorithm such knowledge would be very hard to obtain (for example, the instance PILOTNOV with $n = 329$ and $k_{\min} = 19$ would hardly be feasible to a brute-force exploration of all ${329 \choose 19}$ possibilities).



\section{Conclusions}\label{sec:C}
One of contributions of this paper is demonstrating a novel way of use of fixed-parameter algorithms where they do not substitute heuristic methods but are used to evaluate them.  As a case study, we considered heuristics for the problem of extracting a maximum-size reflected network in an LP problem.  The use of fixed-parameter algorithms helped us to investigate a state-of-the-art heuristic, improve it and allowed us to arrive at an interesting observation that the improved version of the considered heuristic almost always returns an optimal solution.

We believe that this way of applying fixed-parameter algorithms can be useful for other problems as well.  One candidate might be the problem of finding whether the given CNF formula has at most $k$ variables so that their removal makes the resulting formula Renameable Horn.  This is called the Renameable Horn deletion backdoor problem and was recently shown FPT \cite{ROJCSS}.  Heuristics for this problem are widely used in modern SAT solvers for identifying a small subset of variables on which an exponential-time branching
is to be performed \cite{KottBdoor}.  Currently it is unclear whether substituting a heuristic approach by the exact fixed-parameter algorithm would result in a better SAT solver.  But even if it is not the case, the exact algorithm can be still of a considerable use for ranking the heuristic techniques, especially as producing small Renameable Horn backdoors is vitally important for reducing the exponential-time impact on the runtime of SAT solvers.


\bigskip{}

\noindent
\textbf{Acknowledgements} We are grateful to Michael Langston and his group for providing us with a vertex cover code.  The work of G. Gutin  was supported in part by a grant from EPSRC.  The work of I. Razgon was supported by Science Foundation Ireland grant 05/IN/I886.

\bibliographystyle{plain}
\bibliography{sgJ}

\end{document}